\documentclass[11pt]{article}
\usepackage{fullpage}
\usepackage{float}
\usepackage{epsfig}
\usepackage{amsmath}
\usepackage{amssymb}
\usepackage{amstext}
\usepackage{amsmath}
\usepackage{amsthm}
\usepackage{xspace}
\usepackage{latexsym}
\usepackage{verbatim}
\usepackage{multirow}
\usepackage{ifthen}
\usepackage{url}
\usepackage{hyperref}
\usepackage[numbers]{natbib}
\usepackage{aliascnt}
\usepackage{times}

\usepackage[ruled,vlined]{algorithm2e}
\renewcommand{\algorithmcfname}{ALGORITHM}
\SetAlFnt{\small}
\SetAlCapFnt{\small}
\SetAlCapNameFnt{\small}
\SetAlCapHSkip{0pt}
\IncMargin{-0.5em}



\newtheorem{theorem}{Theorem}%

\newaliascnt{lemma}{theorem}
\newtheorem{lemma}[lemma]{Lemma}%
\aliascntresetthe{lemma}

\newaliascnt{claim}{theorem}
\newtheorem{claim}[claim]{Claim}%
\aliascntresetthe{claim}

\newaliascnt{corollary}{theorem}
\newtheorem{corollary}[corollary]{Corollary}%
\aliascntresetthe{corollary}

\newaliascnt{proposition}{theorem}
\aliascntresetthe{proposition}

\renewcommand{\algorithmautorefname}{Algorithm}

\theoremstyle{definition}
\newtheorem{definition}{Definition}

\newtheorem{example}{Example}

\newaliascnt{remark}{theorem}
\newtheorem{remark}[remark]{Remark}
\aliascntresetthe{remark}

\renewcommand{\comment}[1]{}

\newcommand{\ex}{\mathop{\operatorname{E}}}
\newcommand{\pr}{\mathop{\operatorname{Pr}}}

\newcommand{\fall}[1]{\underbar{r}}

\newcommand{\A}{{\mathcal A}}
\newcommand{\M}{{\mathcal M}}

\newcommand{\alg}{\A}
\newcommand{\mech}{\M}

\newcommand{\eps}{{\epsilon}}

\newcommand{\ialg}{\overline{\alg}}

\newcommand{\val}{v}
\newcommand{\vals}{{\mathbf \val}}
\newcommand{\vali}[1][i]{{\val_{#1}}}
\newcommand{\valmi}[1][i]{{\val_{-#1}}}

\newcommand{\alloc}{{x}}
\newcommand{\intalloc}{x}
\newcommand{\allocs}{{\mathbf \alloc}}
\newcommand{\intallocs}{{\mathbf \intalloc}}
\newcommand{\alloci}[1][i]{{\alloc_{#1}}}
\newcommand{\intalloci}[1][i]{{\intalloc_{#1}}}

\newcommand{\dist}{F}
\newcommand{\dists}{{\mathbf \dist}}
\newcommand{\disti}[1][i]{{\dist_{#1}}}

\newcommand{\dens}{f}

\newcommand{\densi}[1][i]{{\dens_{#1}}}

\newcommand{\price}{p}

\newcommand{\pricei}[1][i]{{\price_{#1}}}

\newcommand{\thresh}{t}

\newcommand{\mass}{\alpha}

\newcommand{\AutoAdjust}[3]{\mathchoice{ \left #1 #2  \right #3}{#1 #2 #3}{#1 #2 #3}{#1 #2 #3} }

\newcommand{\InBrackets}[1]{\AutoAdjust{[}{#1}{]}}

\newcommand{\Ex}[2][]{\operatorname{\mathbf E}_{#1}\InBrackets{#2}}
\newcommand{\Prx}[2][]{\operatorname{\mathbf{Pr}}_{#1}\InBrackets{#2}}

\newcommand{\dd}{\mathrm{d}}  

\newcommand{\RED}{\mathcal R}

\newcommand{\levelset}{S}
\newcommand{\winset}{\levelset}

\newcommand{\disalloc}{{\overline \intalloc}}
\newcommand{\disallocs}{{\overline \intallocs}}
\newcommand{\disalloci}{{\disalloc_i}}
\newcommand{\estalloc}{{\tilde \intalloc}}
\newcommand{\estallocs}{{\tilde \intallocs}}
\newcommand{\estalloci}{{\estalloc_i}}

\DeclareMathOperator{\Var}{Var}

\begin{document}
\title{Cost-Recovering Bayesian Algorithmic Mechanism Design
\thanks{This work was done when the first, third and last authors were interning at Microsoft Research New England Lab.}
}
\author{Hu Fu\thanks{Computer Science Dept., Cornell University. \tt{hufu@cs.cornell.edu}.}
\and Brendan Lucier\thanks{Microsoft Research New England. \tt{brlucier@microsoft.com}}
\and Balasubramanian Sivan\thanks{Computer Sciences Dept., University of Wisconsin-Madison. \tt{balu2901@cs.wisc.edu}.}
\and
Vasilis Syrgkanis\thanks{Computer Science Dept., Cornell University. \tt{vasilis@cs.cornell.edu}.}}
\date{}
\maketitle{}

\begin{abstract}
We study the design of Bayesian incentive compatible mechanisms in single parameter domains, for the objective of optimizing social efficiency as measured by social cost.  In the problems we consider, a group of participants compete to receive service from a mechanism that can provide such services at a cost.  The mechanism wishes to choose which agents to serve in order to maximize social efficiency, but is not willing to suffer an expected loss: the agents' payments should cover the cost of service in expectation. 
 
We develop a general method for converting arbitrary approximation algorithms for the underlying optimization problem into Bayesian incentive compatible mechanisms that are cost-recovering in expectation.  In particular, we give polynomial time black-box reductions from the mechanism design problem to the problem of designing a social cost minimization algorithm without incentive constraints.  Our reduction increases the expected social cost of the given algorithm by a factor of~$O(\log(\min\{n, h\}))$, where $n$ is the number of agents and $h$ is the ratio between the highest and lowest nonzero valuations in the support.  We also provide a lower bound illustrating that this inflation of the social cost is essential: no BIC cost-recovering mechanism can achieve an approximation factor better than $\Omega(\log(n))$ or $\Omega(\log(h))$ in general. 

Our techniques extend to show that a certain class of truthful algorithms can be made cost-recovering in the non-Bayesian setting, in such a way that the approximation factor degrades by at most $O(\log(\min\{n, h\}))$.  This is an improvement over previously-known constructions with inflation factor $O(\log n)$.

\end{abstract}

\thispagestyle{empty}

\newpage



\section{Introduction}
\label{sec:intro}

Consider the following scenario: $n$ self-interested agents wish to receive service from a central service provider.  The provider can give service to any set $S$ of the agents, but at a cost $C(S)$, where costs are monotone: $C(S) \leq C(T)$ when $S \subseteq T$.  Each agent has a private value for obtaining service, which they could misrepresent if they so choose.  The provider must decide, given the reported values of the agents, which subset to serve and how much payment to collect from each one.  The goal of the service provider is to maximize the social welfare: the value of the served agents minus the service costs.  
How should the server proceed, given that the agents are rational and may strategically manipulate their declarations?  

If we ignore computational considerations, this mechanism design problem can be resolved via the well-known VCG mechanism, which optimizes social welfare and induces truth-telling as a dominant strategy (i.e., it is in each agent's best interest to report his value truthfully, regardless of the behavior of the other agents).  
If we ignore the incentive constraints, then for many problems of this form (e.g.\ steiner tree, vertex cover, etc.) there are known approximation algorithms that obtain nearly efficient outcomes; however, such algorithms 
in general do not admit payment schemes that would induce truth-telling behavior from the participants.
Finding satisfactory solutions that overcome both the algorithmic and economic difficulties inherent in such problems is the primary research agenda in the field of algorithmic mechanism design.  

A recent line of work has sought to address such problems by considering the \emph{Bayesian} setting, where agent values
are drawn independently from publicly-known distributions.  In such settings, there exist black-box reductions that
convert an arbitrary algorithm into an incentive compatible (i.e., truthful) mechanism with no loss in expected social welfare
\citep{HL10, BH11, HKM11} (where truthfulness in the Bayesian setting means that truth-telling is a Bayes-Nash equilibrium of the mechanism).
Such transformations reduce the mechanism design problem to a purely algorithmic one, decoupling the economic and computational constraints.  A mechanism designer is therefore free to design approximation algorithms, tailored to the specifics of the problem at hand, without paying heed to issues of agent incentives.

Our study begins with the observation that these black-box reductions have an unfortunate property: the server may incur a net loss in expectation.  That is, the payments collected by the mechanism may not cover service costs, in expectation over the agent types. Even a server who wishes to maximize the social welfare may balk at the prospect of following such a protocol.  Our motivating question, then, is whether the theory of Bayesian black-box reductions can be modified to avoid such expected losses.
This can be viewed as a Bayesian version of a \emph{cost-sharing} mechanism design problem, in which the costs for service must be divided among the participants in the mechanism.  Our contribution is to initiate the study of such cost-sharing problems in the Bayesian domain, and to exhibit general black-box reductions converting arbitrary algorithms into truthful cost-sharing mechanisms.

We note that, as in the theory of cost-sharing, one immediately encounters strong impossibility results in such
problems: social welfare is an ill-behaved optimization metric for which no approximation guarantees are possible in
polynomial time even in the full information setting \citep{FPS01}.  Thus, following recent developments in the cost-sharing literature
\citep{RS09}, we describe economic efficiency with respect to minimizing the \emph{social cost}: the service costs plus the total value of the agents who are not served.  

The problem of designing cost-sharing mechanisms that minimize social cost has been extensively studied in the
non-Bayesian domain.  Truthful cost-recovering mechanisms have been developed for many specific problem formulations,
such as Steiner tree/forest \citep{JV01,RS09,GKLRS07,PT03}, facility location \citep{PT03}, multicast routing
\citep{FPS01}, and scheduling problems \citep{BS07}.  These mechanisms generally follow a high-level approach due to
\citet{Mou99}.
Roughly speaking, a Moulin mechanism proceeds by selecting an initial allocation and then iteratively offering cost-recovering prices to the current set of players.  Any player who is not willing to pay his offered price is then removed from the set and the process repeats.  Such mechanisms have been used with great success for numerous problems, but in general the offered prices must be tailored to a particular problem and algorithm; the construction does not generally apply to arbitrary approximation methods.  

A more general construction, also based upon Moulin mechanisms, was recently proposed by \citet{GS12} in the non-Bayesian
setting.  They show that an arbitrary approximation mechanism that is dominant strategy truthful and satisfies a
no-bossiness condition can be converted into a truthful cost-recovering mechanism, while increasing the social cost by a
factor of $O(\log n)$.  This dependency on $n$ matches a lower bound due to \citet{DMRS08}.  While their method applies to many types of algorithms, including a broad class of LP-based algorithms, the truthfulness and no-bossiness requirements limit its generality.  We ask: is a fully general reduction possible in the Bayesian domain, where incentive and efficiency constraints are required to hold in expectation over the agent types?

\paragraph{Our Results}

Our main result is a general reduction that converts an arbitrary \emph{algorithm} into a Bayesian incentive compatible
mechanism with the property that the server does not incur an expected loss.  
Our reductions are \emph{black-box}, meaning that they require only the ability to query the given algorithm on arbitrary input profiles.
We actually provide two different reductions, with slightly different guarantees on social cost.  The
first increases the expected social cost of the original algorithm by a factor of $O(\log n)$, and the second increases the
social cost of the original algorithm by a factor of $O(\log (\val_{\max} / \val_{\min}))$ where $\val_{\max}$ and $\val_{\min}$ are the largest and smallest non-zero values in the support of the value distributions.  Combining these two constructions, we 
contain the increase in social cost to a factor of~$O(\min\{\log n, \log (v_{\max}/v_{\min})\})$.  

We also demonstrate that the increase in social cost exhibited by our constructions is essential.  Specifically, 
based on the construction of \citet{DMRS08}, we show
that no BIC mechanism that recovers cost in expectation can achieve an approximation factor (to the optimal social cost) better than 
$O(\log (\val_{\max} / \val_{\min}) - \sqrt{\val_{\max} / \val_{\min}n})$.  An implication of this bound is that our dependencies on $n$ and $v_{\max}/v_{\min}$ are tight:
no cost-recovering BIC mechanism can achieve approximation factor $o(\log n)$ or approximation factor $o\left(\log (v_{\max}/v_{\min})\right)$.

The ideas underpinning our 
reductions
are motivated by the Moulin mechanism. We apply the paradigm of determining appropriate payments for each agent, and then repeatedly excluding agents who are unwilling to pay the required amount.  However, rather than sequentially excluding agents from an outcome returned by the algorithm, we apply a pre-processing step to the given algorithm in which we sequentially exclude \emph{potential agent declarations}.  This analysis makes use of well-known characterizations of Bayesian incentive compatibility with respect to an algorithm's interim allocation curves (the expected allocation to a player, as a function of his declaration, over the space of declarations of the other agents).  The result of this pre-processing step will be a pre-computed threshold, specific to the given algorithm; any agent who bids below the threshold will be denied service regardless of the original algorithm's outcome.  High thresholds allow the mechanism to charge large payments, but may substantially increase social costs.  We prove that this tension can be balanced so that costs are recovered but yet the social costs are not increased by too much.  

A technical difficulty in the above approach is that knowledge about the algorithm, necessary to determine appropriate thresholds, must be obtained via sampling, which introduces errors.  In order to guarantee that the mechanism recovers costs entirely, rather than only approximately, it is necessary to modify our mechanisms to recover more cost than strictly necessary.  We prove that this has only a small impact on social cost, which can be made arbitrarily small via additional sampling.

We also note that our mechanism with approximation factor $O(\log (\val_{\max}/\val_{\min}))$ extends to the non-Bayesian
setting as well.  Indeed, we show that the cost-sharing construction due to \citet{GS12} can be modified
so that it increases the social cost of a given algorithm by a factor of $O(\min\{\log n, \log (\val_{\max} / \val_{\min})\})$, rather than $O(\log n)$.  This provides an improvement to the obtainable approximation factors when agent values lie in a small range or are drawn from a small set of possible values.

\paragraph{Related work}

Moulin mechanisms were proposed by \citet{Mou99} and \citet{MS01}, who show that the resulting mechanism will be
cost-recovering as long as the prices offered satisfy a cross-monotonicity condition.  Moulin mechanisms have been
applied to various cost-sharing problems, such as Steiner tree/forest \citep{JV01,RS09,GKLRS07,PT03}, facility location
\citep{PT03}, multicast routing \citep{FPS01}, and scheduling problems \citep{BS07}.  For the most part such mechanisms
are also required to be approximately budget balanced, meaning that the mechanism does generate (too large of) a profit.
\citet{IMM08} showed that, for certain problems, such cross-monotonic pricing methods can imply that the budget-balance approximability factor can be very high.

\citet{RS09} suggested social cost as a metric for social efficiency, allowing the study of approximate efficiency in
cost-recovering mechanisms.  Subsequent work considered the approximation factors of cost-sharing methods according to
this metric, for various problems \citep{RS09,BS07,RS07,CRS06}.
\citet{DMRS08} show that, for the public-excludable good problem ($C(S) = 1$ for all $S \neq \emptyset$, $C(\emptyset) = 0$), any (ex post) truthful cost-recovering mechanism will be an $\Omega(\log n)$ approximation to the optimal social cost.

Georgiou and Swamy provide a general method for converting truthful algorithms into truthful cost-recovering mechanisms.  They say an algorithm $\alg$ is \emph{no-bossy} if, for each $i$, if $\alg$ serves a set $S \ni i$ on input $\vals$, then $\alg$ will also serve this same set $S$ on any input $(\vali', \valmi)$ with $\vali' > \vali$.  They show that any $\rho$-approximate algorithm that is dominant strategy truthful and no-bossy can be converted into a truthful cost-recovering $O(\rho \log n)$-approximate mechanism $\mathcal{M}$.  They also provide a linear programming technique for constructing truthful no-bossy algorithms.  Their reduction applies in the ex post (non-Bayesian) setting, rather than the Bayesian setting that we consider.  

In the Bayesian domain, where truthfulness is relaxed to Bayesian incentive compatibility, there are black-box
reductions that convert approximation algorithms into truthful mechanisms in single-parameter \citep{HL10} and
multi-parameter \citep{BH11, HKM11} domains.  These reductions incur an additive loss to the expected social welfare of the original algorithm, which can be made arbitrarily small.  
These constructions do not consider the cost recovery properties of the resulting mechanisms.


\section{Preliminaries}
\label{sec:prelim}
%
%

\paragraph{Single Parameter Mechanism Design}  
Mechanism design studies optimization problems with private information.  Among a set of bidders $[n] = \{1, 2, \cdots,
n\}$, a mechanism decides upon a
subset~$S$ of receivers of a certain service.  Each bidder~$i$ has a private valuation ~$\vali$ for the service.  To incentivize bidders to reveal their valuations truthfully, the mechanism
also charges a payment.  Formally, a mechanism consists of an \emph{allocation} rule $\alloc: \mathbb R_+^n \mapsto [0,
1]^n$ and a \emph{payment} rule $\price: \mathbb R_+^n \mapsto \mathbb R_+^n$. For a valuation profile $\vals = (\val_1, \ldots, \val_n)$,
$\alloci(\vals)$ is the probability that bidder~$i$ receives the service, and $\pricei(\vals)$ is the
payment made by bidder~$i$.  Bidder~$i$ has a utility of $\alloci(\vals) \vali - \pricei(\vals)$.  A mechanism is said
to be \emph{individually rational} (IR) if no bidder ever has a negative utility.  We impose the IR condition throughout the
paper.  A mechanism is said to be \emph{ex post incentive compatible} or \emph{truthful} if,
\begin{align}
\label{eq:IC}
\tag{IC}
\alloci(\vali, \val_{-i}) \vali - \pricei(\vali, \val_{-i}) \geq \alloci(\vali', \val_{-i}) \vali - \pricei(\vali',
\val_{-i}), \qquad \forall i, \forall \vali, \vali', \val_{-i}.
\end{align}


\paragraph{Social Welfare and Social Cost}
A well studied objective in mechanism design is the \emph{social welfare}, defined as $\sum_{i \in S} \vali$, where
$S$ is the set of bidders receiving a service.  In this work, we focus on scenarios where a cost~$C(S)$ is incurred when
subset~$S$ is served.  We assume that $C(\emptyset) = 0$, and that $\forall S \subset T$, $C(S) \leq C(T)$.
The \emph{social cost} of a subset~$S$ is $C(S) + \sum_{i \notin S} \vali$.  Given an algorithm $\alg: \mathbb R_+^n \mapsto
2^{[n]}$, we write the social welfare of~$\alg$ on~$\vals$ as $SW(\alg, \vals)$, and the social cost similarly as
$SC(\alg, \vals)$.  We say a mechanism recovers its cost if for all $\vals$
$\sum_i \pricei(\vals) \geq C(S)$.  

\paragraph{Bayesian Mechanism Design}
This paper focuses on situations in which bidders have only incomplete information regarding the other bidders,  captured by the study of
Bayesian mechanism design.  Each bidder's valuation~$\vali$ is independently drawn from a known distribution~$\disti$,
with probability density function~$\densi$.
By scaling all values and costs down, we may assume without loss of generality that all distributions are supported on
$[0, 1]$.  We denote by $\val_{\max}$ the supremum of the support of all $\disti$'s, and $\val_{\min}$ the infimum of
nonzero values in the support.  We assume $\val_{\min}$ is bounded away from~$0$ and denote by $h$ the ratio
$\val_{\max} / \val_{\min}$.

The allocation rule of a mechanism gives rise to an interim allocation for each bidder.  The interim allocation
$\intalloci(\vali)$ is bidder~$i$'s probability of getting served, taking an expectation over the other bidders'
valuations, i.e.,  $\Ex[\val_{-i}]{\alloci(\vali, \val_{-i})}$.  A mechanism is said to be \emph{Bayesian incentive
compatible} (BIC) if
\begin{align}
\label{eq:BIC}
\tag{BIC}
\intalloci(\vali) \vali - \Ex[\valmi]{\pricei(\vali, \valmi)} \geq \intalloci(\vali') \vali -
\Ex[\valmi]{\pricei(\vali', \valmi)}, \qquad \forall i, \vali, \vali'.
\end{align}
The term $\Ex[\valmi]{\pricei(\vali, \valmi)}$ is also called the interim payment~$\pricei(\vali)$.  
Interim allocation rules and payments of BIC mechanisms are characterized by the following classic result.

\begin{lemma}[\citealt{Mye81}]
\label{lem:Mye81}
A mechanism with interim allocation rules $\intallocs$ is BIC iff each $\intalloci$ is monotone non-decreasing, and the
expected (or interim) payment~$\pricei(\vali)$ made by bidder~$i$ with
valuation~$\val_i$ is $\val_i \intalloci(\val_i) - \int_0^{\val_i} \intalloci(y) \: \dd y$.
\end{lemma}

We will denote by $SC(\alg)$ the expected social cost of an algorithm~$\alg$, i.e., $\Ex[\vals]{C(\alg(\vals)) + \sum_{i
\notin \alg(\vals)} \vali}$.  We say a mechanism~$\mech$ recovers its cost in expectation if $\Ex[\vals]{\sum_i \pricei(\vals)}
\geq \Ex[\vals]{C(\mech(\vals))}$, where $\mech(\vals)$ is the set of bidders $\mech$ serves at valuation profile~$\vals$,
which is allowed to be randomized.  Under the requirement of cost recovery in expectation, it is without loss of generality to assume
that a BIC mechanism charges a payment of $\tfrac{\pricei(\vali)}{\intalloci(\vali)}$ to bidder~$i$ with
valuation~$\vali$ when he is served, and $0$ when he is not served.


\paragraph{Black-Box Reductions for BIC Mechanisms}
While ex post truthful mechanisms optimize objectives (such as social welfare) in the straitjacket of incentive
constraints, \citet{HL10} showed that, if one were to relax the solution concept to that of BIC, essentially any
(approximate) social welfare maximization algorithm can be transformed to a BIC mechanism with little loss of social
welfare.

\begin{theorem}[\citealt{HL10}]
\label{thm:HL10}
In any single-dimensional setting where the agents' valuations are drawn independently from known distributions, given any~$\eps > 0$, there is a polynomial time computable reduction~$\RED$ such that, given any algorithm~$\alg$, $\RED(\alg)$ is a BIC mechanism with $\ex_\vals[SW(\RED(\alg))] \geq \ex_{\vals}[SW(\alg)] - \eps$.
\end{theorem}

The resampling technique of \autoref{thm:HL10} easily applies to the settings where expected social cost is the
objective, and we have the following corollary.

\begin{corollary}
\label{cor:HL10}
In any single-dimensional setting where the agents' valuations are drawn independently from known distributions, given any $\eps > 0$, there is a polynomial time computable reduction~$\RED$ such that, given any algorithm~$\alg$, $\RED(\alg)$ is a BIC mechanism with $\ex_\vals[SC(\RED(\alg))] \leq \ex_{\vals}[SC(\alg)] + \eps$.
\end{corollary}

We note that \citeauthor{HL10} prove \autoref{thm:HL10} by first showing how to construct an $\epsilon$-BIC mechanism, then showing how to convert this into a BIC mechanism at a small loss to social welfare.  Due to the difference in the objective, the technique does not directly apply to our case; a minor modification to the construction of \citeauthor{HL10} is required.  In
Appendix~\ref{sec:eps-bic} we briefly describe the construction involved.

\section{Bayesian Incentive Compatible Cost Recovery}
\label{SEC:BIC}


In this section we present our main result of converting an arbitrary algorithm to a Bayesian Incentive Compatible
mechanism that approximately minimizes social-cost in expectation, and recovers cost in expectation. We give two
separate reductions, one that gives a $O(\log h)$ approximation, and the other gives an $O(\log n)$ approximation, and thus we have a $O\big(\min\left\{\log h, \log n\right\}\big)$ approximation.  

\begin{remark}
Our reductions require computing certain expectations, which can be obtained
only via sampling, and hence with small error.  In this section, all
expectations are assumed to be accurately available, i.e., we assume functional
access to the interim allocation rules and valuation distributions. We refer to
this as functional access to an algorithm. We address the error due to sampling
in the more realistic black-box model in \autoref{SEC:SAMPLE}.  Also, the
truthful payment corresponding to the algorithm~$\alg$ requires knowing the
output of $\alg$ on infinitely many values.  In this section, we also assume
that interim payments are accurately available, and we describe the procedure
to circumvent this issue in \autoref{SEC:SAMPLE}. 
\end{remark}

\noindent In the presentation of \autoref{thm:bic_general}, it is convenient to scale valuations to $[1, h]$ from $[0, 1]$, mapping $\val_{\min}$
to~$1$.

\begin{theorem}\label{thm:bic_general}
Given functional access to an algorithm~$\alg$ which incurs an expected 
social cost of $C(\alg)$, and when the values of all agents are distributed in $[1,h]$,
the reduction in \autoref{alg:bic-logh} outputs a BIC mechanism, which incurs an expected social
cost of $O(\log h) C(\alg)$ and recovers the cost in expectation.
\end{theorem}

\begin{proof}
\begin{algorithm}[t]
\SetKwInOut{Input}{Input}\SetKwInOut{Output}{Output}
\Input{A BIC algorithm $\alg$ and a valuation profile $\vals$}
\Output{A set of agents to be served, and a price for each agent}
\BlankLine
\nl \label{algstep:interim} Let $\levelset(\vals)$ = set of winners returned by $\alg$ on input $\vals$; compute the  interim allocation rule
$\intallocs$ in~$\alg$ \; 
\nl \label{algstep:findk} \For {$j$ = $0$ to $1+\lfloor\log h\rfloor$}{
\begin{align*}
\text{Let } \levelset_j(\vals) &= \{i \in \levelset(\vals) | \vali \geq 2^j \}\\
\text{Let } \price_{i,j}(\vali) &= \begin{cases} 
0 & \text{if $\vali < 2^j$} \\
\vali \intalloci(\vali) - \int_{2^j}^{\vali} \intalloci(y) \: \dd y & \text{if $\vali \geq 2^j$}
\end{cases}
\end{align*}
\If {$\Ex[\vals]{C(\levelset_j(\vals)} \leq \Ex[\vals]{\sum_i\price_{i, j}(\vali)}$}{
Set $k = j$\;
Go to step~\ref{algstep:output}.
}
}
\nl \label{algstep:output} Serve agents in $\levelset_k(\vals)$, and charge agent $i$ a price of
$\frac{\price_{i,k}(\vali)}{\intalloci(\vali)}$. 
\caption{A 
reduction from BIC cost-recovering-in-expectation social cost
minimization to BIC social cost minimization} 
\label{alg:bic-logh}
\end{algorithm}
By Theorem 3.1 of~\cite{HL10} (arxiv version), 
it is without loss of generality to assume that the input
algorithm~$\alg$ is BIC, i.e., has a monotone increasing interim allocation rule. 
(Theorem 3.1 of~\cite{HL10} is a version of~\autoref{cor:HL10} where full functional
access to allocation rule is available, and shows that in such settings
there is not even an additive $\epsilon$ loss.) 
\autoref{alg:bic-logh} proceeds in two phases. In
the preprocessing phase, it computes a number $k$ to be used in the next phase.  This phase does not
depend on the agents' actual valuations, and only uses information of the valuation distributions and the
algorithm~$\alg$.  In the second phase, the mechanism uses bidders' bids (declarations of their valuations) and~$k$ to
modify the set $\winset(\vals)$ returned by~$\alg$.  The actual set of winners and payments are determined in the second
phase.

In the preprocessing phase, the mechanism experiments with \emph{truncating} the interim allocation of~$\alg$ at
different thresholds.  To truncate an interim allocation rule $\intalloci$ at a threshold is to refuse service to all
agents that report values below the threshold, while keeping intact services to others.  The
resulting interim allocation rule is still monotone after truncation.  The payment
$\price_{i,j}(\vali)$ computed by the algorithm is simply the expected payment made by the agent in such a
truncated allocation rule (recall \autoref{lem:Mye81}). 

By the procedure outlined for computing $k$ in \autoref{alg:bic-logh}, it
follows that the resulting mechanism recovers the cost in expectation, i.e.,
when the ``if'' condition in the algorithm becomes true, cost is recovered in
expectation. (Note that there is always a $k$ for which the ``if'' condition
becomes true: at $k = 1+\lfloor \log h \rfloor$, we have $S_k(\vals) =
\emptyset$, and the ``if'' condition becomes true.) In addition, for all
$j \in 0\dots k-1$, we have
\begin{align}
\label{eq:smallj}
\Ex[\vals]{\sum_i \price_{i,j}(\vali)} < \Ex[\vals]{C(S_j(\vals))}.
\end{align}

We claim that, in expectation, the additional social cost incurred by the mechanism when dropping the agents in~$\levelset(\vals) \setminus
\levelset_k(\vals)$ is bounded by an $O(\log h)$ factor times $C(\alg)$.  To begin with, the expected social cost of the mechanism is 
\begin{align*}
\Ex[\vals]{C(\levelset_k(\vals)) + \sum_{i \notin \levelset(\vals)} \vali\alloci(\vals) + 
\sum_{i\in \levelset(\vals) \setminus \levelset_k(\vals)} \vali}.
\end{align*}  
Since $\levelset_k(\vals) \subseteq \levelset(\vals)$, the first two terms are upper bounded by $C(\alg)$.  We therefore need only to bound the last term.
Looking at this term from each agent's perspective, we have
\begin{align*}
\Ex[\vals]{\sum_{i\in \levelset(\vals) \setminus \levelset_k(\vals)} \vali} = \sum_i \sum_{j=0}^{k-1}
\int_{2^j}^{2^{j+1}} \vali \intalloci(\vali) \densi(\vali)\,\dd\vali \leq 2 \sum_i \sum_{j=0}^{k-1} \int_{2^j}^{2^{j+1}}
2^{j} \intalloci(\vali) \densi(\vali)\, \dd\vali.
\end{align*}
Since $\price_{i,j}(\vali) = \vali \intalloci(\vali) - \int_{2^j}^{\vali} \intalloci(y) \: \dd y \geq 2^j \intalloci(\vali)$ for
$\vali \geq 2^j$, we have
\begin{align*}
 \sum_i \sum_{j=0}^{k-1} \int_{2^j}^{2^{j+1}} 2^{j} \intalloci(\vali) \densi(\vali)\, \dd\vali \leq \sum_{j=0}^{k-1}
\Ex[\vals]{\sum_i \price_{i,j}(\vali)}.
\end{align*}

But by \eqref{eq:smallj}, for each $j < k$, $\Ex[\vals]{\sum_i \price_{i, j}(\vali)}$ is in turn bounded by
$\Ex[\vals]{C(\levelset_j(\vals))} \leq \Ex[\vals]{C(\levelset(\vals))} \leq C(\alg)$.  As there are only $k \leq \log h$
such $j$'s,  the additional social cost is at most $O(\log h)$ times $C(\alg)$.

\end{proof}

We now present the other reduction that gives better approximations when $n$ is smaller than~$h$. 

\begin{theorem}
\label{thm:bic_general_logn}
Given functional access to an arbitrary algorithm~$\alg$ which incurs an expected 
social cost of $C(\alg)$, the reduction in \autoref{alg:bic-logn} outputs a BIC mechanism, which incurs an expected social
cost of $O(\log n) C(\alg)$ and recovers the cost in expectation.
\end{theorem}
\begin{proof}
\begin{algorithm}[t]
\SetKwInOut{Input}{Input}\SetKwInOut{Output}{Output}
\Input{A BIC algorithm $\alg$; $C(S)$ for every set $S$ of agents; a value $\delta > 0$}
\Output{A set of agents to be served, and a price for each agent}
\BlankLine
\nl \label{algstep:interim-logn} Initialize $\levelset_0(\vals) = \levelset(\vals)$ = set of winners returned by $\alg$ on input $\vals$, calculate interim allocation rule
$\intallocs$ in~$\alg$\; 
\nl \label{algstep:findk-logn} \For {$j$ = $0, 1, 2, \dotsc$}{
\begin{align*}
\text{Let } \thresh_j &= \left\lceil\frac{\Ex[\vals]{C(\levelset_{j-1}(\vals))}}{\Ex[\vals]{|\levelset_{j-1}(\vals)|}}\right\rceil_{\delta}, \text{ where $\thresh_0 = 0$ }\\
\text{Let } \levelset_j(\vals) &= \{i \in \levelset(\vals) | \vali \geq \thresh_{j} \}\\ 
\text{Let } \price_{i,j}(\vali) &= \begin{cases} 
0 & \text{if $\vali < \thresh_j$} \\
\vali \intalloci(\vali) - \int_{\thresh_j}^{\vali} \intalloci(y) \: \dd y & \text{if $\vali \geq \thresh_j$}
\end{cases}
\end{align*}
\If {$\Ex[\vals]{C(\levelset_j(\vals)} \leq \Ex[\vals]{\sum_i\price_{i, j}(\vali)}$}{
Set $k = j$\;
Go to step~\ref{algstep:output-logn}.
}
}
\nl \label{algstep:output-logn} Serve agents in $\levelset_k(\vals)$, and charge agent $i \in \levelset_k(\vals)$ a price of
$\tfrac{\price_{i,k}(\vali)}{\intalloci(\vali)}$.
\caption{A 
reduction from BIC cost-recovering-in-expectation social cost
minimization to BIC social cost minimization} 
\label{alg:bic-logn}
\end{algorithm}
The idea behind the reduction is similar to that of \autoref{thm:bic_general}.
The main difference is in the definition of the sets $\levelset_j(\vals)$.  They are defined inductively, 
as sets of agents whose value is above the average cost threshold. 

By the same argument as before, the mechanism recovers the cost in expectation
by its definition of~$k$. In addition, we show that the algorithm will
terminate after $O(1/\delta)$ steps. If for some $j$ we had $t_j\leq
t_{j-1}$ then by definition of $t_j$ we have: 
\begin{align*}
\frac{\Ex[\vals]{C(\levelset_{j-1}(\vals))}}{\Ex[\vals]{|\levelset_{j-1}(\vals)|}}\leq\left\lceil\frac{\Ex[\vals]{C(\levelset_{j-1}(\vals))}}{\Ex[\vals]{|\levelset_{j-1}(\vals)|}}\right\rceil_{\delta}=t_{j}\leq t_{j-1},
\end{align*}
which in turn implies:
\begin{align*}
\Ex[\vals]{C(\levelset_{j-1}(\vals))}\leq t_{j-1}\Ex[\vals]{|\levelset_{j-1}(\vals)|}\leq \Ex[\vals]{\sum_i\price_{i, j-1}(\vali)},
\end{align*}
where the last inequality follows from noting that at iteration $j-1$ the
payment of any player that is served is at least $t_{j-1}$. Thus at that point
it must be that the algorithm terminated at iteration $j-1$, since the ``if"
condition gets satisfied. Therefore, it follows that as long as the algorithm
has not terminated, $t_j>t_{j-1}$, and since by definition the thresholds are multiples of
$\delta$ \footnote{$\lceil x \rceil_{\delta}$ is the smallest multiple of $\delta$ that is larger than $x$} it must be that $t_j\geq t_{j-1}+\delta$.
Thus the algorithm will terminate after $O(1/\delta)$ steps (since values lie
in $[0,1]$).

To complete the proof, similar to \autoref{thm:bic_general}, we need to bound the term
$\Ex[\vals]{\sum_{i \in \levelset(\vals) \setminus \levelset_k(\vals)} \vali}.$
Define $r(l)$ inductively as follows: $r(0) = 0$, and 
\begin{align*}
r(j) = \begin{cases}\min \left\{\ell: \Ex[\vals]{|\levelset_{r(j-1)}(\vals)| -|\levelset_{\ell}(\vals)|} \geq 1\right\} & \text{if such an $\ell$ ($\leq k$) exists}\\
k & \text{otherwise}\end{cases}
\end{align*} 
Note that 
$r(j) > r(j-1)$, and $\thresh_{r(j)} > \thresh_{r(j-1)}.$
Let $j_{\max}$ be the smallest $j$ for which $r(j) = k$. Note that $j_{\max} \leq n$ since there are at most $n$ agents, i.e., 
$\Ex[\vals]{|\levelset(\vals)|} \leq n$ and therefore the number of times that the expected service set size can decrease is at most $n$. 
Let $\mass_j = \Ex[\vals]{|\levelset_{r(j-1)}(\vals)| - |\levelset_{r(j)}(\vals)|}$. By definition $\mass_j \geq 1$ for all $j < j_{\max}$. 
We have, 
\begin{align*}
\Ex[\vals]{\sum_{i \in \levelset(\vals) \setminus \levelset_k(\vals)} \vali} 
&\leq \sum_i \sum_{j = 1}^{j_{\max}}
\int_{\thresh_{r(j-1)}}^{\thresh_{r(j)}} \vali \intalloci(\vali) \densi(\vali) \: \dd \vali \\
&\leq \sum_i \sum_{j = 1}^{j_{\max}}
\int_{\thresh_{r(j-1)}}^{\thresh_{r(j)}} \thresh_{r(j)}\intalloci(\vali) \densi(\vali) \: \dd \vali \\
&\leq  \sum_{j=1}^{j_{\max}}\thresh_{r(j)} \Ex[\vals]{|\levelset_{r(j-1)}(\vals)| -|\levelset_{r(j)} (\vals)|}\\
&= \sum_{j=1}^{j_{\max}}\thresh_{r(j)}\mass_j \leq \sum_{j=1}^{j_{\max}}\left\lceil\frac{\Ex[\vals]{C(\levelset_{r(j) -1
}(\vals))}}{\Ex[\vals]{|\levelset_{r(j) - 1}(\vals)|}}\right\rceil_{\delta}\alpha_j\\
&\leq \sum_{j=1}^{j_{\max}}\left[\frac{\Ex[\vals]{C(\levelset_{r(j) -1
}(\vals))}}{\Ex[\vals]{|\levelset_{r(j) - 1}(\vals)|}}+\delta\right]\alpha_j\\
&\leq  \sum_{j=1}^{j_{\max}}\frac{\Ex[\vals]{C(\levelset(\vals))}}{\Ex[\vals]{|\levelset(\vals)|} -
\sum_{\ell<j} \mass_\ell}\mass_j + n\delta
\leq O(\log n)\Ex[\vals]{C(\levelset(\vals))} + n\delta, \notag
\end{align*}
where the last-but-one inequality follows from noting that
$\sum_{j=1}^{j_{\max}} \mass_j \leq \Ex[\vals]{|\levelset(\vals)|} \leq n$. 

For the last inequality, begin by noting that $j_{\max} \leq n$ and for all $j<j_{\max}$, $\mass_j \geq 1$. 
We need to show that $\sum_{j=1}^{j_{\max}}\frac{\mass_j}{\Ex[\vals]{|\levelset(\vals)|} -
\sum_{\ell<j} \mass_\ell} \leq O(\log n)$. Note that $\sum_{j=1}^{j_{\max}}\frac{\mass_j}{\Ex[\vals]{|\levelset(\vals)|} -
\sum_{\ell<j} \mass_\ell} \leq
\sum_{j=1}^{j_{\max}}\frac{\mass_j}{\sum_{j\leq\ell\leq j_{\max}} \mass_\ell } \leq 1 + \sum_{j=1}^{j_{\max}-1}\frac{\mass_j}{\sum_{j\leq\ell\leq j_{\max}-1} \mass_\ell }$. In the final summation, all the $\alpha_j$'s involved are at least $1$. The following claim completes the proof of the last inequality since we have $\sum_{j=1}^{j_{\max}} \mass_j \leq n$. 
\begin{claim}\label{cl:log-n}
Given $k$ real numbers $a_1,\dots, a_k$, such that $a_i \geq 1$ for all $i$, 
$$\sum_{j=1}^{k} \frac{a_j}{\sum_{t\geq j} a_t} \leq 2\cdot H_{\sum_{j=1}^{k}\lfloor a_j \rfloor },$$
where $H_r = \sum_{i=1}^{r} \frac{1}{i} \leq 1 + \log r$. 
\end{claim}
\begin{proof}
{\allowdisplaybreaks\begin{align*}
\sum_{j=1}^{k} \frac{a_j}{\sum_{t\geq j} a_t} &= \sum_{j=1}^{k} \frac{\lfloor a_j\rfloor}{\sum_{t\geq j} a_t} + 
\sum_{j=1}^{k} \frac{a_j - \lfloor a_j\rfloor}{\sum_{t\geq j} a_t}\\
&\leq \sum_{j=1}^{k} \frac{\lfloor a_j\rfloor}{\sum_{t\geq j} \lfloor a_t \rfloor } + \sum_{j=1}^{k} \frac{1}{\sum_{t\geq j} a_t}\\
&\leq \sum_{j=1}^{k} \frac{\lfloor a_j\rfloor}{\sum_{t\geq j} \lfloor a_t \rfloor } + H_k \\
&\leq \sum_{j=1}^{k} \frac{\lfloor a_j\rfloor}{\sum_{t\geq j} \lfloor a_t \rfloor } + H_{\sum_{j=1}^{k}\lfloor a_j\rfloor}
\end{align*}}

We now show that $\sum_{j=1}^{k} \frac{\lfloor a_j\rfloor}{\sum_{t\geq j} \lfloor a_t \rfloor } \leq H_{\sum_{j=1}^{k}\lfloor a_j\rfloor}.$
We drop the floors, and assume that the $a_j$'s are integers in the part below. Consider the term $\frac{a_j}{\sum_{t\geq j} a_t}.$
{\allowdisplaybreaks\begin{align*}
\frac{a_j}{\sum_{t\geq j} a_t} &= \underbrace{\frac{1}{a_j+a_{j+1}+\ldots+a_{k}}+\ldots+\frac{1}{a_j+a_{j+1}+\ldots+a_{k}}}_{a_j \text{ times}}\\
&\leq \frac{1}{1+a_{j+1}+\ldots+a_{k}}+\frac{1}{2+a_{j+1}+\ldots+a_{k}}+\ldots+\frac{1}{a_j+a_{j+1}+\ldots+a_{k}}\\
&= \sum_{t=1+\sum_{k> j}a_k }^{a_j+\sum_{k> j}a_k}\frac{1}{t}
\end{align*}}

So we have,
\begin{align*}
\sum_{j=1}^{k} \frac{a_j}{\sum_{t\geq j} a_t} \leq \sum_{j=1}^{k} \sum_{t=1+\sum_{k> j}a_k}^{\sum_{k\geq j}a_k}\frac{1}{t} = \sum_{t=1}^{\sum_j a_j}\frac{1}{t} = H_{\sum_j a_j}
\end{align*}
\end{proof}

On choosing $\delta = \epsilon/n$, we approximate social cost to a factor of $O(\log n)$ with an additive loss of $\epsilon$. 
The number of iterations is at most $\frac{1}{\delta} = \frac{n}{\epsilon}$ because after $1/\delta$ iterations
the threshold would have reached $1$. 
\end{proof}

\section{Sampling and the Black-Box Model}
\label{SEC:SAMPLE}

The mechanisms of \autoref{SEC:BIC} work under the assumption that the mechanism designer has complete knowledge of the interim allocation rules and valuation distributions in functional form, and can perform arbitrary calculus on those functions.  This is a strong assumption; in general it may be highly non-trivial to precisely determine the interim allocation rules of an arbitrary algorithm.  In this section we describe ways to implement the reductions in \autoref{SEC:BIC} in a more realistic model: the algorithm $\alg$ is provided as a black box that can be queried on arbitrary input vectors.  We refer to this as the \emph{black-box} model of computation.

Our approach will be to estimate the allocation rules of $\alg$ via sampling, then apply the reductions from the ideal model.  This introduces sampling error that must be bounded; the result will be a mechanism that is \emph{approximately} cost-recovering.  We will then show how to modify our constructions to be cost-recovering in the non-approximate sense.  The following theorem summarizes the result.

\begin{theorem}
\label{thm:bic.blackbox}
Given $\epsilon > 0$, black-box access to algorithm~$\alg$ and distribution $\dists$, 
one can construct a BIC mechanism $\mathcal{M}$ that is cost-recovering in expectation, 
with $\ex_{\vals}[SC(\mathcal{M})] \leq O(\min\{\log(h), \log(n)\}) \ex_{\vals}[SC(\alg)] + \epsilon$.
The mechanism runs in time polynomial in $1/\eps$, $n$, and the runtime of $\alg$.
\end{theorem}

\subsection{Computing BIC payments}
\label{sec:sample-price}
Suppose that we are given an algorithm $\alg$ with monotone interim allocation rules $\intallocs$, and moreover we are
told that charging the expected payments of \autoref{lem:Mye81}, $\pricei(\vali) = \vali\intalloci(\vali) -
\int_0^{\vali} \intalloci(y)\,\dd y$, would recover costs in expectation.  In this case, all that would be required 
to obtain a BIC mechanism is to execute algorithm $\alg$ and compute payments so that the expected
payment of agent~$i$ is $\price_i(\vali)$.  However, in the black-box model the mechanism can determine the value of the allocation rules (and hence the required payments) only approximately; charging approximate payments is insufficient for Bayesian incentive compatibility.

There is a well-known procedure to estimate integrals via random sampling, used by \citet{APTT03} to compute payments.  For the purpose of having a self-contained exposition, we explain the procedure below. 
\begin{theorem}
Let $\dens(\cdot)$ be the probability density function of a random variable $Y \in [0,v]$. Then
$\Ex[Y]{\tfrac{h(Y)}{\dens(Y)}} = \int_{0}^{v} h(z)\, \dd z$. 
\end{theorem}
The proof of the theorem follows from the definition of a probability density function.  Thus one way to estimate the
integral $\int_0^{\vali} \intalloci(y)\, \dd y$ is to draw a random variable $Y$ from the uniform distribution
$U[0,\vali]$, and return $\vali \intalloci(Y)$. In expectation, this quantity precisely equals
$\int_0^{\vali}\intalloci(y)\, \dd y$. Furthermore, the payment of $\vali\intalloci(\vali) - \vali\intalloci(Y)$ is
always non-negative since $\intalloci(\cdot)$ is monotone, and thus the mechanism is ex-post IR. 


\subsection{Estimating Interim Allocation via Sampling}
\label{sec:sample-intalloc}

We now describe a method for implementing Algorithms~\ref{alg:bic-logh} and~\ref{alg:bic-logn} when the interim
allocation rules are not given explicitly.  Recall first that, by \autoref{cor:HL10}, given $\epsilon_1 > 0$ and an
arbitrary algorithm $\alg$, we can construct an algorithm $\ialg$ with monotone interim allocation rules such that
$\Ex[\vals]{SC(\ialg)} \leq \Ex[\vals]{SC(\alg)} + \eps_1$.  We will therefore assume for the remainder of this section
that the algorithm~$\alg$ has monotone interim allocation rules.

Given black-box access to algorithm $\alg$, we will construct approximations to its allocation rules as follows.  We
first choose some $\delta > 0$, and partition value space $[0,1]$ into intervals $I_k = ((k-1)\delta, k\delta]$ for $k
\in [1/\delta]$.  Let $\disallocs$ denote the interim allocation rule $\intallocs$, \emph{discretized} over
the intervals $I_k$: that is, for each $\vali \in I_k$, we define $\disalloci(\vali) = \Ex[\vali]{\intalloci(\vali) | \vali \in I_k}$. 

For each $i$ and each $k \in [1/\delta]$, we will sample $N = \frac{1}{\epsilon^2}\log(nk/\epsilon)$ valuation profiles
$\vals \sim \dists$, conditional on $\vali \in I_k$.  We will then run $\alg$ on each of these $N$ inputs, and count the
number of times that the resulting allocation includes agent~$i$.  Denote by $M_{ik}$ this number.  Let
$\estallocs$ be the allocation rule defined by $\estalloci(\vali) = \max_{\ell \leq k}\{M_{i\ell}\}/N$
for all $\vali \in I_k$.  We think of $\estallocs$ as an estimated version of $\intallocs$.  Note that the
reason for the $\max$ in the definition of $\estalloci$ is to guarantee that $\estallocs$ is monotone.

We claim that the result of this sampling generates an estimate to allocation rule $\allocs$, in the following sense.

\begin{definition}
\label{def:eps-close}
Allocation rules $\intallocs$ and $\intallocs'$ are \emph{$\eps$-close} if $|\intalloci(\vali) - \intalloci'(\vali)| < \eps$ for all $i$ and $\vali$.
\end{definition}

\begin{lemma}
\label{lem.sampling}
Let $\estallocs$ be the allocation rule defined by $\estalloci(\vali) = M_{ik}/N$.  Then with probability $1-\eps$ over
the randomness in the sampling procedure, $\estalloci$ is $\eps$-close to $\disalloci$ for all $i$.
\end{lemma}

Once our sampling is complete, we have full functional access to curves $\estallocs$.  We can therefore apply
Algorithms~\ref{alg:bic-logh} and~\ref{alg:bic-logn} to the curves $\estallocs$.  We claim that, for either algorithm,
the analysis of \autoref{SEC:BIC} will go through unchanged, except that each mechanism will be only approximately cost
recovering.  We obtain the following result, the proof of which appears in Appendix~\ref{app:sample-intalloc}.

\begin{lemma}
\label{lem:sample-intalloc}
Given $\eps > 0$ and black-box access to algorithm~$\alg$, 
one can construct a BIC mechanism $\mech$ with $\Ex[\vals]{SC(\mech)} \leq O(\min\{\log(h), \log(n)\})
\Ex[\vals]{SC(\alg)} + \eps$.
Moreover, the expected payments in $\mech$ are at least $\Ex[\vals]{C(\winset(\vals))} - \eps$.
\end{lemma}

It remains to show that we can modify our mechanism to recover expected costs entirely, rather than approximately.  This
requires a modification to Algorithms \ref{alg:bic-logh} and~\ref{alg:bic-logn}.  Each algorithm is currently
designed to iterate until $\Ex[\vals]{\sum_i \pricei(\vals)} \geq \Ex[\vals]{C(\winset(\vals))}$.  We will modify each
algorithm to instead iterate until $\Ex[\vals]{\sum_i \pricei(\vals)} \geq \Ex[\vals]{C(\winset(\vals))} + \eps_0$, for some appropriate $\eps_0 > 0$.  This additional payment of $\epsilon_0$ will be chosen to cover the expected losses due to sampling error.  
What we must show is that this modification does not inflate the expected social cost by too much.
However, this follows immediately from the form of our analysis: in either case, our analysis proceeds by bounding the
loss with respect to the chosen threshold, then bounding this threshold with respect to $\Ex[\vals]{C(\winset(\vals))}$.
If we replace this latter bound with $(\Ex[\vals]{C(\winset(\vals))} + \eps)$, the result is an extra term that is at most $\eps n$.  
An appropriate choice of $\eps$ therefore leads to an arbitrarily small increase to the social cost, and the expected
sum of payments is at least $\Ex{C(\winset_{\estallocs}(\vals))} + \eps \geq ( \Ex{C(\winset_{\allocs}(\vals))} -
\eps ) + \eps = \Ex{C(\winset_{\intallocs}(\vals))}$, as required.  The resulting mechanisms therefore recover costs in
expectation, completing the proof of \autoref{thm:bic.blackbox}.

\section{Lower Bound for BIC Expected Cost-Recovering Mechanisms} 
\label{sec:bic-lb}

We now show that a lower bound on the approximation to social cost given by \citet{DMRS08} extends
to BIC mechanisms, and we tighten the analysis there so that the lower bound is in terms of both $n$ and~$h$.  
In particular, if $h<n$, then the lower bound is $\Omega(\log h)$.  In general, we show a lower bound of $\Omega(\log h -
\sqrt{h / n})$.

\begin{example}[Lower Bound on BIC Social Cost Minimization with Cost Recovery]
\label{ex:lb}
Consider the following public excludable good problem: agent~$i$'s valuation is $v_i = \frac{a_i}{4n}$ where each $a_i$
is drawn independently according to the so-called equal revenue distribution with density function $\dens(z)=\frac{1}{z^2}$ for $z\in [1,h]$ and is $0$ with probability $\frac{1}{h}$.   The cost function is given by $C(\emptyset) = 0$ and $C(S) = 1$ for all $S \neq \emptyset$.
\end{example}

It is easy to see that, without requiring cost recovery, we may simply serve every agent and incur a cost of~$1$.  Next
we give a lower bound for the expected social cost of any cost recovering BIC mechanism.  The proof of
\autoref{THM:BIC-LB} is deferred to Appendix~\ref{app:bic-lb}.

\begin{theorem}
\label{THM:BIC-LB}
Any BIC mechanism for the public excludable good problem described above that recovers cost in expectation has expected social cost at least $\Omega(\log(h)-\sqrt{\frac{h}{n}})$.
\end{theorem}

\section{Ex-post Truthful Cost Recovery}
\label{SEC:EXPOST}
\citet{GS12} proposed the following notion of \emph{no bossiness} for an algorithm and gave a procedure to convert any
truthful, no-bossy algorithm to an ex post truthful, cost recovering mechanism with an inflation of social cost up to a
factor of $O(\log n)$.

\begin{definition}[No Bossy, \cite{GS12}]
\label{def:no-bossy}
An algorithm~$\alg$ is said to be \emph{no bossy} if, for every $i$, $\vali$, $\vali'$ and $\valmi$, if $i \in
\winset(\vali, \valmi)$ and $i \in \winset(\vali', \valmi)$, then $\winset(\vali, \valmi) = \winset(\vali', \valmi)$.
\end{definition}

In this section, we show that such a conversion is also possible with an inflation of $O(\log h)$ in social cost.  For
the special case in which all agents have either value $0$ or~$1$, our conversion does not require the input algorithm
to be either truthful or no bossy.  Proofs from this section appear in Appendix~\ref{app:expost}


\subsection{Black-box reduction for $0/1$ valuations}
\label{sec:ep-01}
When all bidders' valuations are either $0$ or~$1$, there is a simple procedure to convert any social cost minimization
algorithm to an ex post truthful, cost recovering mechanism without increasing the social cost, as we show in
\autoref{alg:ep-01} and \autoref{thm:ep-01}.

\begin{algorithm}[t]
\SetKwInOut{Input}{Input}\SetKwInOut{Output}{Output}
\Input{An algorithm $\alg$ and a valuation profile $\vals$}
\Output{A set of agents to be served, and a price for each agent}
\BlankLine
\nl Initialize $\winset(\vals)$ = set of winners returned by $\alg$ on input $\vals$\; 
\nl Let $\widehat{\winset}(\vals) \leftarrow \winset(\vals) \setminus Z(\vals)$, where $Z(\vals) = \{i\in
\winset(\vals)| \vali = 0\}$\;
\nl \If {$C(\widehat{\winset}(\vals)) \leq |\widehat{\winset}(\vals)|$}{
 Serve agents in $\widehat{\winset}(\vals)$; charge a price of $1$ for each agent in $\widehat{\winset}(\vals)$ and zero for the rest.
}
\Else{Don't serve any agent and charge zero}
\caption{A black-box reduction from ex-post truthful cost-recovering social cost
minimization to social cost minimization for $0/1$ valuations} 
\label{alg:ep-01}
\end{algorithm}

\begin{theorem}
\label{thm:ep-01}
When bidders have only valuations $0$ or~$1$, given black-box access to an arbitrary algorithm~$\alg$ which incurs a
social cost of~$C(\alg)$, the black-box reduction in \autoref{alg:ep-01} outputs an ex post truthful mechanism whose
social cost is no more than $C(\alg)$.
\end{theorem}


\subsection{Black-box reduction for general valuations}
\label{sec:ep-gen}
In this section, for convenience of presentation we again scale up the valuations so that they lie in the range $[1,
h]$.  We give a black-box conversion from a truthful, no bossy mechanism to a truthful, cost recovering mechanism with
an inflation of social cost by a factor of $O(\min \{\log h, \log n\})$.  This is achieved by choosing the better one between
\autoref{alg:ep-gen} and the reduction by \citet{GS12}, whose inflation factor is bounded by $O(\log n)$ alone.  We now
show that the inflation factor of \autoref{alg:ep-gen} is bounded by $O(\log h)$.



\begin{algorithm}[t]
\SetKwInOut{Input}{Input}\SetKwInOut{Output}{Output}
\Input{A truthful, no-bossy mechanism $\mech$, and a valuation profile $\vals$}
\Output{A set of agents, and a payment for each of them}
\BlankLine
\nl Initialize $\winset(\vals)$ = set returned by $\mech$ on input $\vals$, and $\levelset_j(\vals) = \{i \in \winset(\vals)| \vali \geq 2^j\}$\; 
\nl \For {$j$ = $0$ to $\lfloor\log h\rfloor$}{
If $2^j\cdot|S_j(\vals)| \geq C(\levelset_j(\vals))$:
\begin{enumerate}
\item set $k = j$
\item Serve agents in $\levelset_k(\vals)$; charge each of them a price of $2^k$
\end{enumerate}
}
\nl Set $k = 1+\lfloor\log h\rfloor$; no agent is served or charged.
\caption{A black-box reduction from truthful cost-recovering social cost
minimization to truthful, no-bossy social cost minimization} 
\label{alg:ep-gen}
\end{algorithm}

\begin{theorem}
\label{thm:ep-gen}
When values of all agents lie in $[1,h]$, given black-box access to a truthful, no-bossy mechanism~$\mech$ with social
cost~$C(\mech)$, 
the black-box reduction in \autoref{alg:ep-gen} outputs a mechanism 
which recovers cost and incurs a social cost of $O(\log h) C(\mech)$.
\end{theorem}

Instead of experimenting with thresholding at powers of~$2$, \autoref{alg:ep-gen} has the option of proceeding at more
flexible paces.  In particular, we easily obtain the following corollary.

\begin{corollary}
Given a truthful, no-bossy mechanism~$\mech$ that incurs a social cost of $C(\mech)$, and when valuations of all agents reside in $\left\{\vali[1],\dots, \vali[k]\right\}$,
there exists an efficiently computable black-box reduction that outputs a mechanism which recovers cost and incurs a
social cost of at most $O(k C(\mech))$.
\end{corollary}

\bibliographystyle{plainnat}
\bibliography{bibs}

\appendix
\section{Improving $\eps$-BIC to BIC}
\label{sec:eps-bic}

In this section we discuss the construction from the statement of Corollary \ref{cor:HL10}.   
%
%
The purpose of the discussion is to illustrate a minor modification to the method of \citealt{HL10} for converting $\eps$-BIC mechanisms to BIC mechanisms.  We will briefly recall their $\eps$-BIC construction for the sake of completeness, then describe how to modify it to obtain a BIC mechanism while incurring only a small increase to the social cost.

\subsection{The $\eps$-BIC Reduction}

The $\eps$-BIC reduction due to \citealt{HL10} is as follows.  Suppose $\alg$ is an arbitrary social cost algorithm with (unknown) interim allocation rules $\allocs$.  For each $i$, we will first partition the value space $[0,1]$ into $1/\delta$ intervals of width $\delta$; write $I_k = (k\delta, (k+1)\delta]$ for the $k$th such interval.  

Suppose first that we knew the value of $\ex_{\vali}[\alloci(\vali)\ |\ \vali \in I_k]$ for each $i$ and $k$.  Given this information, we could perform the following monotonizing operation.  We construct a certain partition $\mathcal{P}$ of $[0,1]$ (into intervals), where $\mathcal{P}$ is a coarsening of the intervals $I_k$; that is, each interval endpoint in $\mathcal{P}$ will be a multiple of $\delta$.  Suppose the intervals in $\mathcal{P}$ are $I_1', \dotsc, I_\ell'$ for some $\ell \leq 1/\delta$.  Given $\mathcal{P}$ (whose construction we have not described), we will define algorithm $\overline{\alg}$ as follows:
\begin{enumerate}
\item For each agent $i$, if $\vali \in I_j' \in \mathcal{P}$, then draw $\vali' \sim \disti$.
\item Return $\alg(\vals')$
\end{enumerate}
\citealt{HL10} show that there is a way to construct partition $\mathcal{P}$ so that $\overline{\alg}$ is BIC, and $\overline{\alg}$ has the same social cost as $\alg$.

Next suppose that the value of $\ex_{\vali}[\alloci(\vali)\ |\ \vali \in I_k]$ is not known explicitly for all $i$ and $k$.  In this case, our reduction will attempt to estimate these values.  We do so by taking many samples $\vals \sim \dists$ subject to $\vali \in I_k$ and executing $\alg$ on each sampled value profile; our estimate for $\ex_{\vali}[\alloci(\vali)\ |\ \vali \in I_k]$ will be the fraction of these samples for which $\alg$ serves agent $i$.  Chernoff-Hoeffding bounds imply that if we take $O(\frac{1}{\lambda^2} \log(n/\lambda\delta))$ samples per agent and interval, then every estimate will be within $\lambda$ of the true value with probability at least $1 - \lambda$.  Write $\tilde{\allocs}$ for the estimated allocation curves.

Estimates in hand, we can perform the monotonizing operation described above on the estimated allocation curves, as if they were the actual curves.  \citealt{HL10} show that, if this is done, the resulting algorithm is approximately monotone: with probability $(1-\lambda)$, it is true that for each $i$, if $\overline{\alloc}_i$ is the interim allocation rule for agent $i$, then $\overline{\alloc}_i(\vali) \leq \overline{\alloc}_i(\vali) + 2\lambda$ for all $\vali \leq \vali'$. Also, the expected social welfare decreases by at most $(\delta + 2\epsilon)n$ as a result of this reduction.  Since social cost is simply $\sum_{i \in [n]} \vali$ minus the social welfare, this implies the social cost increases by at most $\epsilon$ for an appropriate choice of $\lambda$ and $\delta$.  

\subsection{Obtaining a BIC Reduction}

Suppose that $\ialg$ is an $\epsilon$-BIC algorithm, constructed via the reduction described above.  An important fact about $\ialg$ is that it has allocation rules that are piecewise-constant on the intervals $I_k$.  Thus, any non-monotonicities in an interim allocation curve of $\ialg$ can occur only at finitely many possible input values; specifically, at multiples of $\delta$.  Our approach for modifying $\ialg$ to be BIC will therefore be to modify its allocation curve in a blatantly monotone way over the intervals between these multiples of $\delta$.  Specifically, we will reduce the probability of allocating to an agent $i$ by $\epsilon k$ whenever he declares a value on interval $k$.  Since there are only $1/\delta$ such intervals, the overall increase in social cost due to this change will be small.

More formally, we perform the following modification to algorithm $\ialg$.  Our new algorithm, $\tilde{\alg}$, proceeds as follows, where we set $\gamma = 2\epsilon/\delta$.
\begin{enumerate}
\item With probability $1-\gamma$, return $\ialg(\vals)$.
\item Otherwise, choose an agent $i$ uniformly at random. If $\vali \in I_k$, then return $\{i\}$ with probability $k\delta$, otherwise return $\emptyset$.
\end{enumerate}

\begin{claim}
If $\ialg$ is $\epsilon$-BIC and piecewise constant on intervals $I_k$, then $\tilde{\alg}$ is BIC.
\end{claim}
\begin{proof}
Write $\tilde{\alloci}$ for the interim allocation rule of $\tilde{\alg}$.  Choose any $1 \leq k < 1/\delta$ and suppose $v \in I_k$ and $v' \in I_{k+1}$.  Then 
\begin{align*}
\tilde{\alloci}(v) & = (1-\gamma)\overline{\alloc}_i(v) + \gamma k\delta \\
& \leq (1-\gamma)(\overline{\alloc}_i(v')+2\eps) + \gamma k\delta \\
& \leq (1-\gamma)\overline{\alloc}_i(v') + 2\eps + \gamma (k+1)\delta - \gamma\delta \\
& = (1-\gamma)\overline{\alloc}_i(v') + \gamma (k+1)\delta \\
& = \tilde{\alloci}(v')
\end{align*}
and hence $\tilde{\alloci}$ is monotone, as required.
\end{proof}

\begin{claim}
The expected social cost of $\tilde{\alg}$ is at most the expected social cost of $\ialg$ plus $\gamma n$.
\end{claim}
\begin{proof}
Note that we can assume that $C(\{i\}) \leq n$ for all $i$; otherwise we would never serve $i$ and could remove agent $i$ from the mechanism.  We therefore have
\begin{align*}
{\ex}_{\vals}[SC(\tilde{\alg},\vals)] 
& \leq (1-\gamma){\ex}_{\vals}[SC(\ialg,\vals)] + \gamma \cdot \max\left\{\sum_i \vali, C(\{1\}), \dotsc, C(\{n\})\right\} \\
& \leq (1-\gamma){\ex}_{\vals}[SC(\ialg,\vals)] + \gamma n
\end{align*}
as required.
\end{proof}

Thus, for $\epsilon' = 2n\epsilon/\delta$, we conclude that $\tilde{\alg}$ is BIC and has expected social cost at most that of $\ialg$ plus $\epsilon'$.  In other words, given an arbitrary $\epsilon'$, we can choose $\epsilon$ and $\delta$ sufficiently small (but polynomial in $\epsilon'$ and $1/n$) such that $\tilde{\alg}$ is BIC and increases expected social cost of the original algorithm $\alg$ by at most $\epsilon'$, as required by \autoref{cor:HL10}.

\section{Omitted proofs from Section~\ref{SEC:SAMPLE}}
\label{app:sample-intalloc}

\subsection{Proof of \autoref{lem.sampling}}

The expected value of $M_{ik}/N$ is precisely $\ex_{\vals}[\alloci(\vals)\ |\ \vali \in I_k]$.  By Chernoff-Hoeffding
bounds, after $N$ samples the probability that $|M_{ik}/N - \Ex{M_{ik}/N}| \geq \eps$ is at most $\eps\delta/n$.  Taking
the union bound over all $i$ and $k$, we have that $|M_{ik}/N - \Ex{M_{ik}/N}| < \eps$ for all $i$ and $k$ with
probability at most $1 - \eps$.  This also implies that $|\estalloci(\vali) - \disalloci(\vali)| < \epsilon$ for all $i$ and $\vali$, 
as required.\qed

\subsection{Proof of \autoref{lem:sample-intalloc}}

Let us first recall the statement of the lemma.  Given $\eps > 0$ and black-box access to algorithm~$\alg$, 
we claim that one can construct a BIC mechanism $\mech$ with $\Ex[\vals]{SC(\mech)} \leq O(\min\{\log(h), \log(n)\})
\Ex[\vals]{SC(\alg)} + \eps$.
Moreover, the expected payments in $\mech$ are at least $\Ex[\vals]{C(\winset(\vals))} - \eps$.

We first note that discretizing allocation curves along intervals of length $\delta$ can increase social cost by at most $\delta n$, as each agent's value changes by at most $\delta$ as a result of this approximation.  We therefore note this increase in social cost and assume for notational convenience that each $\alloci$ is constant on each interval $I_k$.

Our mechanism will proceed by constructing $\estallocs$ as described above, and then apply either \autoref{alg:bic-logh}
or \autoref{alg:bic-logn}, depending on which of $n$ or $h$ is smaller.  In either case, the algorithm will compute some
threshold~$T$.  The mechanism will proceed by eliciting valuation profile~$\vals$, querying~$\alg(\vals)$, and then
serving those agents in the resulting set $\winset$ with value at least $T$.  Regardless of the threshold returned, this mechanism will be BIC.

In the event that $\estallocs$ is not $\eps$-close to $\intallocs$, the social cost generated by the resulting mechanism
is trivially bounded by~$n$.  Since this event occurs with probability at most $\eps$, its contribution to the expected
social cost is at most $\eps n$.  We therefore assume that $\estallocs$ and $\intallocs$ are $\eps$-close.

For either algorithm, the threshold~$T$ is chosen so that expected payments, as computed from $\estallocs$, recover
expected costs in expectation.  Since each $\estalloci$ is $\eps$-close to the true curve $\intalloci$, the estimated
payments differ from the true payments by at most $\eps$, for each agent.  This has two effects: first, in our analysis
of each algorithm, bounds on the increase to social cost include an error of up to $\eps$ per agent, as
$\estalloci(\vali) \leq \intalloci(\vali)+\eps$ for all~$i$.  Thus, there can be up to an additional $\eps n$ increase
in social cost due to the threshold $T$ applied.

Second, the true BIC payment (from \autoref{lem:Mye81}) may differ from the approximate BIC payment by up to $\eps$, for
each agent.  Thus, the expected payments of mechanism $\mech$ may be up to $\eps n$ less than was computed by either
algorithm.  In particular, it may be that the expected payments are as low as $\Ex[\vals]{C(\winset(\vals))} - \eps n$.

To summarize, our resulting mechanism will have $\Ex[\vals]{SC(\mech)} \leq O(\min\{\log(h), \log(n)\})
\Ex[\vals]{SC(\alg)} + (\delta + 2\eps) n$, and will have expected payments at least $\Ex[\vals]{C(\winset(\vals))} - \eps n$.  Taking an appropriate choice of $\eps$ and $\delta$ then completes the proof.
\qed

\section{Proof of Theorem~\ref{THM:BIC-LB}}
\label{app:bic-lb}

It is well known in auction theory (e.g. \citealt{Mye81}) that, from an agent whose valuation is drawn
from the equal revenue distribution, a BIC mechanism can extract a payment of at most~$1$ in expectation.  Therefore,
any BIC mechanism for \autoref{ex:lb} can collect payments at most~$\tfrac 1 4$ in expectation.  Since the mechanism recovers
cost in expectation, the expected cost must be at most $\tfrac{1}{4}$ as well.  But unless $\winset = \emptyset$,
$C(\winset) = 1$.  Therefore, 
\begin{align*}
\pr[\winset \neq \emptyset] \leq \frac{1}{4}.
\end{align*}

Let $V$ be $\sum_i \vali$.  Observe that 
\begin{align*}
\Ex{V} = n \Ex{\vali} = \frac{\log h}{4}, \quad
\Var[V] = n \Var[\vali] \leq \frac{h}{16n}, \quad
\sigma(V) \leq \frac{1}{4} \sqrt{h/n}. 
\end{align*}
By Chebyshev's inequality, we have
\begin{align*}
\Prx{V < \frac{\log(h) - 2\sqrt{h/n}}{4}} \leq \frac 1 4.
\end{align*}
The expected social cost of a cost recovering BIC mechanism~$\mech$ is at most:
\begin{align*}
\Ex{SC(\mech)}\geq \frac{\log(h) - 2\sqrt{h/n}}{4} \Prx{\winset = \emptyset \wedge V \geq \frac{\log(h) -
2\sqrt{h/n}}{4}}.
\end{align*}
By the union bound, the latter probability is at least $\tfrac 1 2$. Therefore,
\begin{align*}
\Ex{SC(\mech)} \geq \frac{\log(h) - 2\sqrt{h/n}}{8}.
\end{align*}
\qed

\section{Omitted Proofs from Section~\ref{SEC:EXPOST}}
\label{app:expost}

\subsection{Proof of Theorem~\ref{thm:ep-01}}

The mechanism output by \autoref{alg:ep-01} can be clearly seen to be truthful: an agent with value~$0$ never
wins, and an agent with value~$1$ gets a zero utility, and so no agent has motivation to misreport his value.
It recovers cost because it serves agents only if the cost can be recovered.  Also, as $\widehat{\winset}(\vals)
\subseteq \winset(\vals)$, if $\widehat{\winset}$ is served, then the social cost is less than that $C(\alg)$ since the
agents in $\winset(\vals) - \widehat{\winset}(\vals)$ does not add to the social cost; on the other hand, if no agents
are served, the change in social cost is $C(\widehat{\winset}(\vals)) - |\widehat{\winset}(\vals)| < 0$. 
\qed

\subsection{Proof of Theorem~\ref{thm:ep-gen}}

It is easy to see that the procedure in \autoref{alg:ep-gen} guarantees cost recovery.  To see that it is truthful, note
that if an agent is served, then misreporting his valuation leads either to non-service (to his disadvantage) or to
service with the same cost and payment (by the no bossiness of~$\mech$); if an agent is not served, he will not have an
incentive to overreport his valuation to be served because he would not do that in~$\mech$ (because $\mech$ is truthful)
and now the payment is even higher than in~$\mech$.  Now, similarly to the proof of \autoref{thm:bic_general}, we need only to bound the additional social cost inflicted by
refusing service to bidders with valuations no more than~$2^k$.

By the way $k$ is determined, we have
\begin{equation}\label{eq:ep-threshold}
2^{j}|\levelset_j(\vals)| < C(\levelset_j(\vals)), \quad \forall j < k.
\end{equation}
Using this, we have
\begin{align*}
\sum_{i\in \winset(\vals)\setminus \levelset_{k}(\vals)} \vali = \sum_{j = 0}^{k - 1} \sum_{i \in \levelset_j(\vals) \setminus
\levelset_{j + 1}(\vals)} \vali \leq \sum_{j = 0}^{k - 1} 2^{j + 1} |\levelset_j(\vals)| \leq \sum_{j = 0}^{k - 1} 2
C(\levelset_j(\vals)) \leq O(\log h) C(\mech).
\end{align*}
\qed

\end{document}